\documentclass[12pt]{amsart}
\usepackage[margin=2.5cm]{geometry}
\usepackage{vaucanson-g}
\usepackage{url}

\newtheorem{theorem}{Theorem}
\newtheorem{proposition}[theorem]{Proposition}
\newtheorem{lemma}[theorem]{Lemma}
\newtheorem{corollary}[theorem]{Corollary}

\theoremstyle{definition}
\newtheorem{definition}[theorem]{Definition}
\newtheorem{remark}[theorem]{Remark}
\newtheorem{example}[theorem]{Example}

\newtheorem{problem}{Problem}

\title{The growth function of $S$-recognizable sets}
\author[\'E. Charlier]{\'Emilie Charlier}
\author[N. Rampersad]{Narad Rampersad}

\address[\'E. Charlier]{David R. Cheriton School of Computer Science, University of Waterloo,
Waterloo, ON N2L 3G1, Canada}
\email[\'E. Charlier]{echarlier@uwaterloo.ca}

\address[ N. Rampersad]{Institute of Mathematics, University of Li\`ege, 
Grande Traverse 12 (B 37), B-4000 Li\`ege, Belgium}
\email[N. Rampersad]{nrampersad@ulg.ac.be}

\DeclareMathOperator{\rep}{rep}
\DeclareMathOperator{\val}{val}
\DeclareMathOperator{\N}{\mathbb N}
\DeclareMathOperator{\Q}{\mathbb Q}
\DeclareMathOperator{\Z}{\mathbb Z}

\begin{document}
\begin{abstract}
A set $X\subseteq\N$ is $S$-recognizable for an abstract numeration system $S$ if the set $\rep_S(X)$ of 
its representations is accepted by a finite automaton. 
We show that the growth function of an $S$-recognizable set is always either $\Theta((\log(n))^{c-df}n^f)$ 
where $c,d\in\N$  and $f\ge 1$, or $\Theta(n^r \theta^{\Theta(n^q)})$, where $r,q\in\Q$ with $q\le 1$. 
If the number of words of length $n$ in the numeration language is bounded by a polynomial, 
then the growth function of an $S$-recognizable set is $\Theta(n^r)$, where $r\in \Q$ with $r\ge 1$. 
Furthermore, for every $r\in \Q$ with $r\ge 1$, we can provide an abstract numeration system $S$ 
built on a polynomial language and an $S$-recognizable set such that the growth function of $X$ is $\Theta(n^r)$. 
For all positive integers $k$ and $\ell$, we can also provide an abstract numeration system $S$ 
built on a exponential language and an $S$-recognizable set such that the growth function of $X$ 
is $\Theta((\log(n))^k n^\ell)$.  
\end{abstract}

\maketitle

\section{Introduction}
A set $X \subseteq \N$ is \emph{$b$-recognizable} if the set of
representations of the elements of $X$ in base $b$ is accepted by a
finite automaton.  The class of $b$-recognizable sets is a very
well-studied class (the main results can be found in the book of
Allouche and Shallit \cite{AllSha}).  It is therefore somewhat
surprising that, to the best of our knowledge, no characterization of
the possible growths of $b$-recognizable sets is currently known.  In
this paper we provide such a characterization for the much more general
class of $S$-recognizable sets.

Let $L$ be a language over an alphabet $\Sigma$. 
We assume that the letters $a_1,\ldots,a_\ell$ of $\Sigma$ are ordered by $a_1<\cdots<a_\ell$. 
This order on the alphabet $\Sigma$ induces an order $<$ on the language $L$ called the {\em genealogic order} 
or the {\em radix order}: words are ordered length by length, and for a given length, we use the lexicographic order. 
This leads to the definition of an abstract numeration system.

\begin{definition}
An \emph{abstract numeration system} is a triple $S = (L, \Sigma, <)$
where $L$ is an infinite language over a totally ordered
finite alphabet $(\Sigma, <)$.  The language $L$ is called the {\em numeration language}.
The map $\rep_S\colon  \N \to L$ is a bijection mapping $n \in \N$ to the $(n+1)$-th word of $L$
ordered genealogically. The inverse map is denoted by $\val_S \colon L \to\N$.
\end{definition}

Most of the time, we assume that $L$ is a regular language. 
Nevertheless, some results hold true for arbitrary numeration languages.

The integer base numeration systems are particular cases of abstract numeration systems since for all $b,n,m\in\N$ with $b\ge2$, we have $n<m\Leftrightarrow \rep_b(n)<\rep_b(m)$, where $\rep_b(x)$ designates the usual greedy representation 
of $x$ in base $b$. In this case, the numeration language is 
$\rep_b(\N)=\{1,2,\ldots,b-1\}\{0,1,2,\ldots,b-1\}^*\cup\{\varepsilon\}$, which is regular.

\begin{definition}
Let $S$ be an abstract numeration system. Let $X\subseteq\N$. 
The set $X$ is \emph{$S$-recognizable} if the language $\rep_S(X)= \{\rep_S(n) \colon n\in X\}$ is regular.
Let $b \geq 2$ be an integer.  
The set $X$ is \emph{$b$-recognizable} if it is $S$-recognizable for the abstract numeration system $S$ built on
the language $\rep_b(\N)$ consisting of the base-$b$ representations of the elements of $X$.
The set $X$ is \emph{$1$-recognizable} if
it is $S$-recognizable for the abstract numeration system $S$ built on $a^*$.
\end{definition}

Note that $\N$ is $S$-recognizable if and only if the numeration language $L$ is regular.

Rigo \cite{Rigo} stated the following two fundamental questions regarding
$S$-recognizable sets:
\begin{itemize}
\item For a given numeration system $S$, what are the $S$-recognizable
  subsets of $\N$?

\item For a given subset $X$ of $\N$, is it $S$-recognizable for some
  numeration system $S$?
\end{itemize}

Of course, both of these questions are quite difficult to address in full
generality.  In this paper, we consider these questions in terms of
the growth functions of both the set $X$ and the numeration
language.

\begin{definition}
For a subset $X$ of $\N$, we let $t_X(n)$ denote the $(n+1)$-th term of $X$. The map
$t_X\colon\mathbb{N}\to\mathbb{N}$ is called the {\em growth function} of $X$.
\end{definition}

We address the following problem.

\begin{problem}
Let $S$ be an abstract numeration system built on a regular language. 
What do the growth functions of $S$-recognizable sets look like? 
Of course, this question has to be answered in terms of the growth function of the numeration language.
\end{problem}

Apart from the following result, known as {\em Eilenberg's gap theorem}, 
we are not aware of any result in the direction of answering the above problem.

\begin{theorem}\cite{Ei}
Let $b\ge2$ be an integer. A $b$-recognizable set $X$ of nonnegative integers satisfies either 
$\limsup_{n\to+\infty} (t_X(n+1)-t_X(n))<+\infty$ or $\limsup_{n\to+\infty} \frac{t_X(n+1)}{t_X(n)}>1$. 
\end{theorem}

Thanks to this result, examples of sets that are not $b$-recognizable
for any $b$ have been exhibited. The set $\{n^2 : n \in \N\}$ of
squares is such an example.  However, the set of squares is
$S$-recognizable for the abstract numeration system
\[
S = (a^*b^* \cup a^*c^*, \{a,b,c\}, a<b<c).
\]
More generally, Rigo \cite{Rigo} and Strogalov \cite{Strogalov} showed
that for any polynomial $P \in \Q[x]$ such that $P(\N) \subseteq \N$,
there exists $S$ such that $P$ is $S$-recognizable.  Observe that in
the case of an integer base numeration system, the number of words of
each length in the numeration language grows exponentially, whereas in
the case of the numeration system $S$, this number grows polynomially.
This leads to the natural question: Can a set of the form $P(\N)$ ever
be recognized in a numeration system where the numeration language is
exponential?  In Section~\ref{sec:additional}, we show that the answer
to this question is no for all polynomials $P$ of degree $2$ or more.

Let us fix some asymptotic notation.

\begin{definition}
Let $f$ and $g$ be functions with domain $\mathbb{N}$. We say that~$f$ is $\Theta(g)$, 
and we write $f=\Theta(g)$, if there exist positive constants $c$ and $d$ and a nonnegative integer $N$ such that, 
for all integers~$n\ge N$, we have $cg(n)\le f(n)\le dg(n)$.  
We say that $f$ and $g$ have {\em equivalent behaviors at infinity}, 
which is denoted by $f(n)\sim g(n) \ (n\to+\infty)$ (or simply $f\sim g$ when the context is clear), 
if we have $\lim_{n\to+\infty} \frac{f(n)}{g(n)}=1$. 
Finally, we write $f=o(g)$ if  if we have $\lim_{n\to+\infty} \frac{f(n)}{g(n)}=0$.
\end{definition}

\begin{definition}\label{def:complexity}
For any language $L$ over an alphabet $\Sigma$ and any non-negative integer~$n$, we let
\[\mathbf{u}_L(n)=|L\cap\Sigma^n|\]
denote the number of words of length $n$ in $L$ and
\[\mathbf{v}_L(n)=\sum_{i=0}^n \mathbf{u}_L(i)=|L\cap\Sigma^{\le n}|\]
denote the number of words of length less than or equal to $n$ in $L$. The maps
$\mathbf{u}_L\colon\mathbb{N}\to\mathbb{N}$ and $\mathbf{v}_L\colon\mathbb{N}\to\mathbb{N}$ are called the {\em counting (or growth) functions of $L$}.
\end{definition}

When $L$ is a regular language, the sequence $(\mathbf{u}_L(n))_{n\ge 0}$ 
satisfies a linear recurrence relation with integer coefficients (for instance, see \cite{berstel-reutenauer}): 
there exist a positive integer $k$ and $a_1,\ldots,a_k\in\Z $ such that we have
\[\forall n\in\N,\ \mathbf{u}_L(n+k)=a_1 \mathbf{u}_L(n+k-1)+\cdots +a_k \mathbf{u}_L(n).\]
Then, since we have $\mathbf{v}_L(n)-\mathbf{v}_L(n-1)=\mathbf{u}_L(n)$, the sequence $(\mathbf{v}_L(n))_{n\ge 0}$ satisfies the linear recurrence relation  of length $k+1$ whose characteristic polynomial is $(x-1)(x^k-a_1x^{k-1}-\cdots-a_k)$.

Our main result can be stated as follows.

\begin{theorem}\label{the:main}
Let  $S=(L,\Sigma,<)$ be an abstract numeration system built on a regular language 
and let $X$ be an infinite $S$-recognizable set of nonnegative integers. Suppose  
\[\forall i\in\{0,\ldots,p-1\},\ \mathbf v_L(np+i)\sim a_in^c\alpha^n \ (n\to+\infty),\]
for some $p,c\in\N$ with $p\ge1$, some $\alpha\ge1$ and some positive constants $a_0,\ldots,a_{p-1}$,
and 
\[\forall j\in\{0,\ldots,q-1\},\ \mathbf v_{\rep_S(X)}(nq+j)\sim b_jn^d\beta^n \ (n\to+\infty),\]
for some $q,d\in\N$ with $q\ge1$, some $\beta\ge1$ and some positive constants $b_0,\ldots,b_{q-1}$.
Then we have 
\begin{itemize}
\item $t_X(n)=\Theta\left((\log(n))^{c-d\frac{\log(\sqrt[p]{\alpha})}{\log(\sqrt[q]{\beta})}} 
n^{\frac{\log(\sqrt[p]{\alpha})}{\log(\sqrt[q]{\beta})}}\right)$ if $\beta>1$;
\item $t_X(n)=\Theta\left(n^{\frac{c}{d}} (\sqrt[p]{\alpha})^{\Theta(n^{\frac{1}{d}})}\right)$ if $\beta=1$.
\end{itemize}
Furthermore we have
\begin{itemize}
\item $t_X(n)=\Theta\left(n^{\frac{c}{d}} (\sqrt[p]{\alpha})^{(1+o(1))(\frac{n}{b})^{\frac{1}{d}}}\right)$  
if $\beta=1,\ q=1$, and $\mathbf v_{\rep_S(X)}(n)\sim bn^d\ (n\to+\infty)$;
\item $t_X(n)=\Theta\left( n^{\frac{c}{d}} (\sqrt[p]{\alpha})^{(\frac{n}{b})^{\frac{1}{d}}}\right)$  
if $\beta=1, \ q=1$, and $\mathbf v_{\rep_S(X)}(n)= bn^d$ for all $n\in\N$.
\end{itemize}
\end{theorem}

This paper has the following organization. 
In Section~\ref{sec:background}, we recall some necessary background concerning $S$-recognizable sets. 
In Section~\ref{sec:proof}, we provide a proof of Theorem~\ref{the:main}. 
In Section~\ref{sec:particular}, we describe some constructions of abstract numeration systems 
$S$ that effectively realize some particular behaviors of $S$-recognizable sets. 
As a consequence of our main result, we also show that certain behaviors of $S$-recognizable 
sets are never achieved by any abstract numeration systems $S$.
Then, in Section~\ref{sec:examples}, we propose examples illustrating 
every possible behavior of $t_X(n)$ when $X$ is an $S$-recognizable set. 
Finally, in Section~\ref{sec:additional}, we discuss a few additional results. 
Namely, we show that no polynomial can be recognized in an abstract numeration system 
built on an exponential regular language. 
This result extends a well-known result in the integer base numeration systems.

\section{Background on $S$-recognizable sets}\label{sec:background}

\begin{definition}
Let $S$ be an abstract numeration system. 
An infinite word~$x$ over an alphabet $\Gamma$ is {\em $S$-automatic} if, 
for all non-negative integers $n$, its $(n+1)$st letter $x[n]$ is obtained 
by ``feeding'' a DFAO $\mathcal{A}=(Q,\Sigma,\delta,q_0,\Gamma,\tau)$ with the $S$-representation of $n$:
\[\forall n\in\N,\ \tau(\delta(q_0,\rep_S(n)))=x[n].\]
\end{definition}

\begin{definition}
Let $X$ be a set of nonnegative integers. Its characteristic sequence 
is the sequence $\chi_X=(\chi_X(n))_{n\ge0}$ defined by \[\chi_X(n)= \left \{ \begin{array}{ll}
1,& \text{if}\ n\in X;\\
0,& \text{otherwise}.
\end{array} \right.\]
\end {definition}

The following result is self-evident.

\begin{proposition}\label{pro:aut}
Let $S$ be an abstract numeration system.
A set is $S$-recognizable if and only if its characteristic sequence is $S$-automatic.
\end{proposition}

\begin{definition}
If $\mu$ is a morphism over an alphabet $\Sigma$ and $a$ is a letter in $\Sigma$ such that 
the image $\mu(a)$ begins with $a$, then we say that $\mu$ is {\em prolongable on $a$}.
\end{definition}

If a morphism $\mu$ is prolongable on a letter $a$, then the limit $\lim_{n\to+\infty}\mu^n(a)$
is well defined. As usual, we denote this limit by $\mu^\omega(a)$. 
Furthermore, this limit word is a fixed point of $\mu$. 
Observe that it is an infinite word if and only if there is a letter $b$ occurring in $\mu(a)$ that satisfies
 $\mu^n(b)\neq\varepsilon$ for all non-negative integers $n$.

\begin{definition}
An infinite word is said to be {\em pure morphic} if it can be written as $\mu^\omega(a)$ 
for some morphism $\mu$ prolongable on a letter $a$. It is said to be {\em morphic} 
if it is the image under a morphism of some pure morphic word.
\end{definition}

\begin{theorem}\cite{R,RM} \label{the:RM}
An infinite word is $S$-automatic for some abstract numeration system $S$ if and only if it is morphic.
\end{theorem}

\begin{example}\label{ex:pansiot}
Consider the morphism $h\colon\{0,1\}^*\to\{0,1\}^*$ defined by $h(1)=1010$ and $h(0)=00$. 
Let $X$ be the set of nonnegative integers whose characteristic sequence is $h^\omega(1)$.
The first element in $X$ are $0,2,6,8,16,18,22,24,40,42,46,\ldots$ From Theorem~\ref{the:RM} 
and Proposition~\ref{pro:aut}, the set $X$ is $S$-recognizable for some abstract numeration system $S$. 
The associated DFAO is depicted in Figure~\ref{fig:nlogn} (details on how to build this DFAO are given in Section~\ref{sec:proof}). 
For all $k\in\N$, we have $|h^k(1)|=(k+1) 2^k$. For all $n\in\N$, there is a unique $k:=k(n)\in\N$ such that 
$(k+1) 2^k\le t_X(n)<(k+2) 2^{k+1}$. Since the number of occurrences of $1$ in the prefix $h^k(1)$ is equal to $2^k$, 
we also have $(k+1) 2^k\le t_X(n)<(k+2) 2^{k+1}\Leftrightarrow 2^k\le n<2^{k+1}$. 
This means that $k(n)=\log_2(n)$. Consequently, $t_X(n)$ is $\Theta(n\log(n))$.
\begin{figure}[htbp]
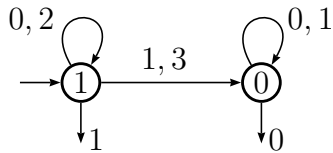

\centering
\VCDraw{%
\begin{VCPicture}{(0,-1.5)(4,2)}
 \State[1]{(0,0)}{1}
 \State[0]{(4,0)}{0}
\Initial[w]{1}
\FinalL{s}{1}{1}
\FinalL{s}{0}{0}
\LoopN{1}{0,2}
\EdgeL{1}{0}{1,3}
\LoopN[0.75]{0}{0,1}
\end{VCPicture}}
    \caption{A DFAO generating $X$.}
\label{fig:nlogn}
\end{figure}
\end{example}

\section{Proof of the main result}\label{sec:proof}

\begin{definition}
Let  $\mathcal{A}=(Q,\{a_1,\ldots,a_r\},\delta,q_0,F)$ be a DFA whose alphabet is ordered by $a_1<\cdots<a_r$. 
The {\em morphism associated with} $\mathcal{A}$ is the morphism 
$\mu_\mathcal{A}\colon (Q\cup\{\alpha\})^*\to (Q\cup\{\alpha\})^*$ defined by 
\[\forall q\in Q,\ \mu_\mathcal{A}(q)=\delta(q,a_1)\cdots \delta(q,a_r)\ \text{ and }\ 
\mu_\mathcal{A}(\alpha)=\alpha\mu_\mathcal{A}(q_0).\]
where $\alpha$ is a new letter, not belonging to $Q$.
\end{definition}

In \cite{R}, it was shown that any $S$-automatic word is morphic (see Theorem~\ref{the:RM} above). 
Here we only need to consider the particular case where the $S$-automatic word is the characteristic word 
of an $S$-recognizable set of nonnegative integers. We follow the construction of \cite{R} in the following definition.
 Nevertheless, notice that our definitions of $\mathcal A$ and $g$ are slightly different.

\begin{definition}\label{def:gF}
Let $S=(L,\Sigma,<)$ be an abstract numeration system and let $X$ be an $S$-recognizable set of nonnegative integers. 
Let $\mathcal{A}_L$ be the trim minimal automaton of $L$ and let $\mathcal{A}_X$ 
be the (complete) minimal automaton of $\rep_S(X)$. 
Now define $\mathcal A=(Q,\Sigma,\delta,q_0,F)$ to be the product of these two automata 
(take only the accessible part) and consider the canonically associated morphism $\mu_\mathcal{A}$. 
Let $g\colon (Q\cup\{\alpha\})^*\to\{0,1\}^*$ be the morphism defined by $g(\alpha)=g(q_0)$ 
and for all states $(p,q)$ of $\mathcal A$,
\[g(p,q)=
\left \{ \begin{array}{ll}
1, & \text{if}\ p  \text{ is final in } \mathcal{A}_L \text{ and } q \text{ is final in } \mathcal{A}_X;\\
0, & \text{if}\ p  \text{ is final in } \mathcal{A}_L  \text{ and } q \text{ is not final in } \mathcal{A}_X;\\
\varepsilon, & \text{if}\ p  \text{ is not final in } \mathcal{A}_L.
\end{array}\right.\]
Furthermore, for each $n\in\N$, we define $F(n)$ to be the number of occurrences of the letter 
$1$ in $g(\mu_\mathcal{A}^n(\alpha))$.
\end{definition}

The following lemma is self-evident.

\begin{lemma}\label{lem:equiv}
Let $S=(L,\Sigma,<)$ be an abstract numeration system and let $X$ be an $S$-recognizable set of nonnegative integers.
With the notation of Definition~\ref{def:gF}, we have
\[\forall n\in\N,\forall k\in\N,\ |g(\mu_\mathcal{A}^k(\alpha))|\le t_X(n)<|g(\mu_\mathcal{A}^{k+1}(\alpha))|
\Leftrightarrow F(k)\le n<F(k+1).\]
\end{lemma}

\begin{lemma}\label{lem:L}
Let  $S=(L,\Sigma,<)$ be an abstract numeration system and let $X$ be an $S$-recognizable set of nonnegative integers.
With the notation of Definition~\ref{def:gF}, we have
\[\forall n\in\N,\ |g(\mu_\mathcal{A}^n(\alpha))|=\mathbf v_L(n) \text{ and }  F(n)=\mathbf v_{\rep_S(X)}(n).\]
\end{lemma}

\begin{proof}
Define $K$ to be the language accepted by the trim automaton $\mathcal A_K$ 
whose graph is the same as $\mathcal A_L$ but  where all states are final. 
Observe that the following inclusions hold: $\rep_S(X)\subseteq L \subseteq K$. 
By construction, for all $n\in\N$, the $(n+2)$-th letter of the infinite word
$\mu_\mathcal{A}^\omega(\alpha)$ is the state in $\mathcal A$ reached by reading the $(n+2)$-th word in $K$. 
Define the morphism $f\colon (Q\cup\{\alpha\})^*\to Q^*$ by$f(\alpha)=f(q_0)$ 
and for all states $(p,q)$ of $\mathcal A$,
\[f(p,q)=
\left \{ \begin{array}{ll}
(p,q), & \text{if}\ p  \text{ is final in } \mathcal{A}_L;\\
\varepsilon, & \text{if}\ p  \text{ is not final in } \mathcal{A}_L.
\end{array}\right.\]
Then, for all $n\in\N$, the $(n+1)$-th letter of the infinite word $f(\mu_\mathcal{A}^\omega(\alpha))$ 
is the state in $\mathcal A$ reached by reading the $(n+1)$-th word in $L$. 
Now define the morphism $h\colon Q^*\to \{0,1\}^*$ be the morphism defined by 
\[h(p,q)=
\left \{ \begin{array}{ll}
1, & \text{if } q \text{ is final in } \mathcal{A}_X;\\
0, & \text{if } q \text{ is not final in } \mathcal{A}_X,
\end{array}\right.\]
for all states $(p,q)$ of $\mathcal A$. We clearly have $g=h\circ f$.

We claim that $|\mu_\mathcal{A}^n(\alpha))|=\mathbf v_K(n)$ for all $n\in\N$. 
If so, since $f$ erases the words not belonging to $L$ and $h$ is a letter-to-letter morphism, 
then we obtain $|f(\mu_\mathcal{A}^n(\alpha))|=|g(\mu_\mathcal{A}^n(\alpha))|=\mathbf v_L(n)$ for all $n\in\N$. 
Then, by definition of $F$, we also obtain $F(n)=\mathbf v_{\rep_S(X)}(n)$ for all $n\in\N$.

It is thus sufficient to prove the claim. 
We define a new automaton $\mathcal B=(Q\cup\{\alpha\},\Sigma\cup\{a_0\},\delta',\alpha,Q\cup\{\alpha\})$ 
by slightly modifying the automaton $\mathcal A$. 
The initial state $q_0$ of $\mathcal A$ is no longer initial in $\mathcal B$ and we add a new initial state $\alpha$ 
with a loop labeled by a new letter $a_0$ (not belonging to $\Sigma$). 
All states are final.
The (partial) transition function $\delta'\colon Q\cup\{\alpha\}\times\Sigma\cup\{a_0\}\to Q\cup\{\alpha\}$ 
is defined by $\delta'(\alpha,a)=\delta(q_0,a)$ for all $a\in\Sigma$, $\delta'(\alpha,a_0)=\alpha$, 
and $\delta'(q,a)=\delta(q,a)$ for all $q\in Q$ and all $a\in\Sigma$. 
Observe that $K$ is the language of the words accepted by $\mathcal B$ 
from which we remove the words starting with $\alpha$. 
With the terminology of \cite{CKR}, this automaton $\mathcal B$ is the {\em automaton canonically associated 
with the morphism $\mu_{\mathcal{A}}$ and the letter $\alpha$} and the corresponding {\em directive language} is $K$. 
We thus can apply \cite[Lemma 28]{CKR} to obtain the claim. This completes the proof.
\end{proof}

\begin{lemma}\label{lem:SS}
For all regular languages $L$, there exist $p,c\in\N$ with $p\ge1$ and $\alpha\ge1$ 
such that for all $i\in\{0,\ldots,p-1\}$, we have 
\[\mathbf v_L(np+i)\sim a_in^c\alpha^n \ (n\to+\infty)\]
where $a_0,\ldots,a_{p-1}$ are some positive constants.
\end{lemma}

\begin{proof}
It is well-known that $\sum_{n\ge0}\mathbf v_L(n)$ is an $\N$-rational series for all regular languages $L$ 
(see for instance \cite{berstel-reutenauer}). 
Since $(\mathbf v_L(n))_{n\ge0}$ is a non-decreasing sequence, 
the lemma follows from \cite[Theorem II.10.2]{Salomaa-Soittola}.
\end{proof}

We are ready for the proof of Theorem~\ref{the:main}.

\begin{proof}[Proof of Theorem~\ref{the:main}]
Let $p,q,a_0,\ldots,a_{p-1},b_0,\ldots,b_{q-1},c,d,\alpha,$  and $\beta$ be numbers like in the statement. 
The hypotheses imply $\mathbf v_L(n)=\Theta(n^c(\sqrt[p]{\alpha})^n)$ and 
$\mathbf v_{\rep_S(X)}(n)=\Theta(n^d(\sqrt[q]{\beta})^n)$. 
Consider the notation of Definition~\ref{def:gF}.  
For all $n\in\N$, there exists a unique $k:=k(n)\in\N$ such that we have 
$\ |g(\mu_\mathcal{A}^k(\alpha))|\le t_X(n)<|g(\mu_\mathcal{A}^{k+1}(\alpha))|$. 
From Lemma~\ref{lem:equiv}, this number $k$ is also the only nonnegative integer that satisfies $F(k)\le n<F(k+1)$. 

First consider the case where $\beta>1$.
From Lemma~\ref{lem:L} we obtain $F(k)=\Theta(k^d(\sqrt[q]{\beta})^k)$.
From \cite[Lemma 4.7.14]{CANT} we find
\[k(n)=\frac{1}{\log(\sqrt[q]{\beta})} (\log(n)-d\log(\log(n)))+O(1).\]
Using Lemma~\ref{lem:L}, this gives the announced asymptotic behavior.

Now consider the case where $\beta=1$.
From Lemma~\ref{lem:L} we obtain $F(k)=\Theta(k^d)$. 
This gives $k(n)=\Theta(n^{\frac{1}{d}})$. 
Hence, from Lemma~\ref{lem:L}, the announced asymptotic behavior holds.
If moreover $q=1$, we can be more precise. 
In this case, we have $F(k)=\mathbf v_{\rep_S(X)}(k)\sim bk^d\ (k\to+\infty)$ for some $b$. 
This gives $k(n)\sim (\frac{n}{b})^{\frac{1}{d}}\ (k\to+\infty)$ and from Lemma~\ref{lem:L}, 
the announced asymptotic behavior holds. 
If moreover $\mathbf v_{\rep_S(X)}(n)=bn^d$ for all $n\in\N$ and for some $b$, then
we have $F(k)=bk^d$ for all $k\in\N$. This gives $k(n)=\lfloor \frac{n}{b}^{\frac{1}{d}}\rfloor$.
Therefore, the announced asymptotic behavior holds in this case as well.
This concludes the proof.
\end{proof}

\begin{remark}\label{rem}
Note that the hypotheses of Theorem~\ref{the:main} implies that either $\sqrt[q]{\beta}<\sqrt[p]{\alpha}$ or ($\sqrt[q]{\beta}=\sqrt[p]{\alpha}$ and $d\le c$). Since $\rep_S(X)$ is assumed to be an infinite language, we also always have $\beta\ge 1$. 
\end{remark}

The following corollary is a particular case of Theorem~\ref{the:main}
where $\alpha=\beta^r$ for some $r\ge1$. Recall that $\alpha$ and
$\beta$ are said to be {\em multiplicatively dependent} if $\alpha=\beta^r$ for some non-null $r\in\Q$.

\begin{corollary}\label{cor:dep}
Let  $S=(L,\Sigma,<)$ be an abstract numeration system built on a regular language 
and let $X$ be an infinite $S$-recognizable set of nonnegative integers. Suppose  
\[\forall i\in\{0,\ldots,p-1\},\ \mathbf v_L(n p+i)\sim a_in^c\beta^{r n} \ (n\to+\infty),\]
for some $p,c\in\N$ with $p\ge1$, some $\beta\ge1$, some $r\ge1$, and some positive constants $a_0,\ldots,a_{p-1}$,
and 
\[\forall j\in\{0,\ldots,q-1\},\ \mathbf v_{\rep_S(X)}(n q+j)\sim b_j n^d\beta^n \ (n\to+\infty),\]
for some $q,d\in\N$ with $q\ge1$ and some positive constants $b_0,\ldots,b_{q-1}$.
Then we have 
\begin{itemize}
\item $t_X(n)=\Theta\left( (\log(n))^{c-d r\frac{q}{p}} n^{r\frac{q}{p}}\right)$ if $\beta>1$;
\item $t_X(n)=\Theta\left( n^{\frac{c}{d}}\right)$ if $\beta=1$.
\end{itemize}
\end{corollary}

\begin{proof}
The case where $\beta>1$ is simply a rewriting of Theorem~\ref{the:main} with $\alpha=\beta^r$ 
and $r=\frac{\log(\alpha)}{\log(\beta)}$. 
For the case $\beta=1$, observe that the hypotheses imply $\beta^r=1$. 
Hence, the conclusion follows directly from Theorem~\ref{the:main} with $\alpha=\beta=1$.
\end{proof}

\section{Achieving particular behaviors}\label{sec:particular}

Recall that a language $L$ is {\em polynomial} if $\mathbf u_L(n)$ is $O(n^d)$ for some $d\in\N$ and 
is {\em exponential} if there exist $c>0$ and $\theta>1$ such that the inequality $\mathbf u_L(n)\ge c\theta^n$ 
holds for infinitely many integers $n$. Let us recall now the following gap theorem.

\begin{theorem}\cite{Sz}
Any regular language is either polynomial or exponential.
\end{theorem}

\begin{proposition}\label{prop:kl}
\hspace{1cm}
\begin{itemize}
\item For all $k,\ell\in \N$ with $\ell>0$, there exists an abstract numeration system $S$ 
built on an exponential regular language and an infinite $S$-recognizable 
set $X\subseteq\N$ such that $t_X(n)=\Theta((\log(n))^k n^\ell)$.
\item For all $k,\ell\in \N$ with $\ell>1$, there exists an abstract numeration system $S$ 
built on an exponential regular language and an infinite $S$-recognizable 
set $X\subseteq\N$ such that $t_X(n)=\Theta\left(\frac{n^\ell}{(\log(n))^k}\right)$.
\item For all positive integer $k$ and for all abstract numeration systems $S$, there is no $S$-recognizable 
set $X\subseteq\N$ such that $t_X(n)=\Theta\left(\frac{n}{(\log(n))^k}\right)$.
\end{itemize}
\end{proposition}

\begin{proof}
Let $\ell$ be a positive integer and let $k\in\N$.

From \cite[Proposition 17]{Rigo}, for all $c\in\N$, there exists a regular language 
$L$ having $\mathbf v_L(n)=(n+1)^c 2^{\ell(n+1)}$ as growth function (for $n\ge1$).
This language $L$ is obtained by considering unions and shuffles of regular languages over distinct alphabets, 
i.e, alphabets having  empty pairwise intersections. 
In particular, for all nonnegative integers $b\le c$, it contains a regular sublanguage $M^{(b)}$ 
such that $\mathbf u_{M^{(b)}}(n)=n^b2^n$ for all $n\in\N$.

Let $S=(L,\Sigma,<)$ be an abstract numeration system built 
on a regular language $L$  satisfying $\mathbf v_L(n)\sim n^k 2^{\ell(n+1)}$ ($n\to+\infty$). 
This language is defined as in the previous paragraph. 
Then the set $X=\val_S(M^{(0)})$ is an infinite $S$-recognizable set
such that $\mathbf v_{\rep_S(X)}(n)\sim 2^{n+1}$ $(n\to+\infty)$. 
From Theorem~\ref{the:main}, we obtain $t_X(n)=\Theta((\log(n))^k n^\ell)$. This proves the first assertion.

Now we assume $\ell\ge 2$.
Choose any integer $d\ge\frac{k}{\ell-1}$ and let $c=\ell d-k$. From the choice of $d$, we have $c\ge d$. 
Let $S=(L,\Sigma,<)$ be an abstract numeration system built 
on a regular language $L$   satisfying $\mathbf v_L(n)\sim n^c 2^{\ell(n+1)}$ ($n\to+\infty$).
This language is defined as in the first paragraph. 
Then the set $X=\val_S(M^{(d)})$ is an infinite $S$-recognizable set
such that $\mathbf v_{\rep_S(X)}(n)\sim n^d 2^{n+1}$ $(n\to+\infty)$. 
From Theorem~\ref{the:main}, we obtain $t_X(n)=\Theta\left(\frac{n^\ell}{(\log(n))^k}\right)$.
This proves the second assertion.

Consider now the third assertion and assume $k>0$. Let $S=(L,\Sigma,<)$ be an abstract numeration system. 
Suppose that such a set $X$ exists.
In view of  Lemma~\ref{lem:SS} and Theorem~\ref{the:main}, we should have
\[\forall j\in\{0,\ldots,q-1\},\ \mathbf v_{\rep_S(X)}(nq+j)\sim b_jn^d\beta^n \ (n\to+\infty),\]
for some $q,d\in\N$ with $q\ge1$, some $\beta>1$ and some positive constants $b_0,\ldots,b_{q-1}$.
In the same way, we can write 
\[\forall i\in\{0,\ldots,p-1\},\ \mathbf v_L(np+i)\sim a_i n^c\alpha^n \ (n\to+\infty),\]
for some $p,c\in\N$ with $p\ge1$, some $\alpha>1$ and some positive constants $a_0,\ldots,a_{p-1}>0$. 
Then we must have $\sqrt[p]{\alpha}=\sqrt[q]{\beta}$. 
From Theorem~\ref{the:main} and Remark~\ref{rem}, 
we then obtain $t_X(n)=\Theta((\log(n))^{c-d}n)$ with $c\ge d$, a contradiction. This ends the proof.
\end{proof}

The following corollary of Theorem~\ref{the:main} considers the case of a polynomial numeration language.

\begin{corollary}\label{cor:poly}
Let $S=(L,\Sigma,<)$ be an abstract numeration system built on a polynomial regular language 
and let $X$ be an infinite $S$-recognizable set of nonnegative integers.
Then we have $t_X(n)=\Theta( n^r)$ for some rational number $r\ge1$.
\end{corollary}

\begin{proof}
Since $L$ is an infinite polynomial regular language, its growth function $\mathbf v_L(n)$ must satisfy 
\[\forall i\in\{0,\ldots,p-1\},\ \mathbf v_L(np+i)\sim a_in^c \ (n\to+\infty),\]
for some $p,c\in\N$ with $p\ge1$ and some positive constants $a_0,\ldots,a_{p-1}$.
The sublanguage $\rep_S(X)$ of $L$ is necessarily polynomial too, and since $X$ is an infinite $S$-recognizable set, 
the growth function $\mathbf v_{\rep_S(X)}(n)$ must satisfy
\[\forall j\in\{0,\ldots,q-1\},\ \mathbf v_{\rep_S(X)}(nq+j)\sim b_jn^d\ (n\to+\infty),\]
for some $q,d\in\N$ with $q\ge1$ and $d\le c$, and some positive constants $b_0,\ldots,b_{q-1}$.
Then from Theorem~\ref{the:main}, we obtain $t_X(n)=\Theta( n^\frac{c}{d})$.
\end{proof}

\begin{proposition}
For every rational number $r\ge1$, 
there exists an abstract numeration system $S$ 
built on a polynomial regular language and an infinite $S$-recognizable 
set of nonnegative integers $X$ such that $t_X(n)=\Theta( n^r)$.
\end{proposition}

\begin{proof}
Fix a rational number $r\ge1$. Write $r=\frac{c}{d}$ where $c$ and $d$ are positive integers. 
Define $\mathcal B_\ell$ to be the {\em bounded language} $a_1^*a_2^*\cdots a_\ell^*$. 
We have $\mathbf v_{\mathcal B_\ell}(n)=\binom{n+\ell}{\ell}$ for all $\ell\ge1$ and $n\in\N$ 
(for example see \cite[Lemma 1]{CRS}). 
Let $S$ be the abstract numeration system built on $\mathcal B_c$ with the order $a_1<a_2<\cdots<a_c$ 
and let $X=\val_S(\mathcal B_d)$ (since $c\ge d$, we have $\mathcal B_d\subseteq \mathcal B_c$). 
Hence we have $\mathbf v_{\mathcal B_c}(n)=\binom{n+c}{c}$ and $\mathbf v_{\rep_S(X)}(n)=\binom{n+d}{d}$ 
for all $n\in\N$. Then from Theorem~\ref{the:main}, we obtain $t_X(n)=\Theta( n^\frac{c}{d})=\Theta( n^r)$.
\end{proof}

\section{Examples}\label{sec:examples}

In this section we provide several examples to illustrate the constructions of Sections~\ref{sec:proof} and~\ref{sec:particular}.

\begin{example}\label{ex:pansiot2}
Let us continue Example~\ref{ex:pansiot}. 
Consider the abstract numeration system $S$ built on the language $L$ accepted by the automaton of Figure~\ref{fig:nlogn} 
from which are removed the words beginning with $0$ and the alphabet order $0<1<2<3$. 
The trim minimal automaton of $L$ is depicted in Figure~\ref{fig:nlogn2}.
\begin{figure}[htbp]
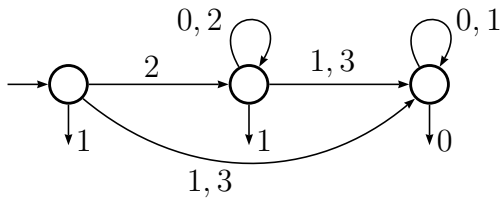

\centering
\VCDraw{%
\begin{VCPicture}{(-4,-2)(4,2)}
 \State{(0,0)}{1}
 \State{(4,0)}{0}
\State{(-4,0)}{i}
\Initial[w]{i}
\FinalL{s}{1}{1}
\FinalL{s}{0}{0} 
\FinalL{s}{i}{1}
\LoopN{1}{0,2}
\EdgeL{1}{0}{1,3}
\LoopN[0.75]{0}{0,1}
\VArcR{arcangle=-40}{i}{0}{1,3}
\EdgeL{i}{1}{2}
\end{VCPicture}}
    \caption{The trim minimal automaton of $L$.}
\label{fig:nlogn2}
\end{figure}
The set $X$ is $S$-recognizable since $\rep_S(X)=2\{0,2\}^*\cup\{\varepsilon\}$. 
We have $\mathbf u_L(n)=(n+2)2^{n-1}$ if $n\ge1$ and $\mathbf u_L(0)=1$ 
and $\mathbf u_{\rep_S(X)}(n)=2^{n-1}$ if $n\ge1$ and $\mathbf u_{\rep_S(X)}(0)=1$. 
This gives $\mathbf v_L(n)=(n+1)2^n$ and $\mathbf v_{\rep_S(X)}(n)=2^n$ for all $n\in\N$.
Observe that accordingly to Lemma~\ref{lem:L}, we have $|h^n(1)|=\mathbf v_L(n)=(n+1)2^n$ 
and $F(n)=\mathbf v_{\rep_S(X)}(n)=2^n$ for all $n\in\N$, where $F(n)$ is the number of occurrences of the letter $1$ in 
$|h^n(1)|$. Recall that we found \[t_X(n)=\Theta(n\log(n)),\] which is consistent with Theorem~\ref{the:main}.
\end{example}

\begin{example}
Consider the base $4$ numeration system, that is, the abstract numeration system built on 
$\mathcal L_4=\{\varepsilon\}\cup\{1,2,3\}\{0,1,2,3\}^*$ with the natural order on the digits.
Let $X=\val_4(L)=\{0,2,8,10,32,34,40,42,128,130,136,138,160,\ldots\}$ where $L$ is the language accepted by the automaton of Figure~\ref{fig:nlogn2}.  It is $4$-recognizable. From Example~\ref{ex:pansiot2} and from Theorem~\ref{the:main}, we obtain 
\[t_X(n)=\Theta\left(\left(\frac{n}{\log(n)}\right)^2\right).\]
\end{example}

It is well known that the set of squares $\{n^2\mid n\in\N\}$ 
is not $b$-recognizable for all integer bases $b\ge2$  (for instance see \cite{Ei}). 
In \cite{Strogalov,Rigo} (also see Theorem~\ref{the:construction} below), 
it was shown that any set of the form $\{n^k\mid n\in\N\}$, with $k\in\N$, is $S$-recognizable for some $S$. 
In those constructions, the exhibited abstract numeration systems are built on polynomial languages.
In the following example, we exhibit $4$-recognizable sets of nonnegative integers having their $n$-th terms in $\Theta(n^k)$ for $k=2$ and some $k\not\in\N$. These considerations have to be compared with Proposition~\ref{prop:kl} above and  Proposition~\ref{prop:poly} in the next section.

\begin{example}
Consider again the base $4$ numeration system.
Let $X=\val_4(\{1,3\}^*)=\{1,3,5,7,13,15,21,23,29,31,\ldots\}$. 
It is clearly $4$-recognizable. We have $\mathbf v_{\mathcal L_4}(n)=4^n$ and 
$\mathbf v_{\{1,3\}^*}(n)=2^{n+1}-1$ for all $n\in\N$. 
From Theorem~\ref{the:main}, we obtain \[t_X(n)=\Theta(n^2).\]
This also illustrates Proposition~\ref{prop:kl}.

Now let $Y=\val_4(\{1,2,3\}^*)=\{0,1,3,5,6,7,9,10,11,13,14,15,21,22,\ldots\}$. 
It is clearly $4$-recognizable. We have $\mathbf v_{\{1,2,3\}^*}(n)=\frac{1}{2}(3^{n+1}-1)$ for all $n\in\N$. 
From Theorem~\ref{the:main}, we obtain \[t_Y(n)=\Theta\left(n^{\frac{\log(4)}{\log(3)}}\right).\]  
This also illustrates Corollary~\ref{cor:dep}.

Define $L_F=\{\varepsilon\}\cup1(0+01)^*$ to be the language of the Fibonacci numeration system 
and let $Z=\val_4(L_F)=\{0,1,4,16,17,64,70,256,257,260,272,273,1024,\ldots\}$. Again it is $4$-recognizable. 
We have $\mathbf v_{L_F}(n)=\frac{5+3\sqrt{5}}{10}(\frac{1+\sqrt{5}}{2})^n
+\frac{5-3\sqrt{5}}{10}(\frac{1-\sqrt{5}}{2})^n$ for all $n\in\N$. 
Therefore, we find 
\[\mathbf v_{L_F}(n)\sim\frac{5+3\sqrt{5}}{10}\left(\frac{1+\sqrt{5}}{2}\right)^n\ (n\to+\infty)\] 
and by Theorem~\ref{the:main}, we obtain \[t_Z(n)=\Theta\left(n^{\frac{\log(4)}{\log(\varphi)}}\right)\]
where $\varphi=\frac{1+\sqrt{5}}{2}$ is the golden ratio.

Now we illustrate Theorem~\ref{the:main} when $p=1$ and $q=2$.
Define $K$ to be the language accepted by the automaton depicted in Figure~\ref{fig:K}. 
\begin{figure}[htbp]
\centering
\VCDraw{%
\begin{VCPicture}{(0,-1.5)(4,1.5)}
 \State{(0,0)}{0}
 \State{(4,0)}{1}
\Initial[w]{0}
\Final[s]{0}
\Final[s]{1}
\ArcL[.5]{0}{1}{1,2,3}
\ArcL[.5]{1}{0}{0,2}
\end{VCPicture}}
    \caption{The trim minimal automaton of $K$.}
\label{fig:K}
\end{figure} 
Let $V=\val_4(K)$. The first values of $V$ are $0, 1, 2, 3, 4, 6, 8, 10, 12, 14, 17, 18, 19, 25, 26,27, 33$.
It is $4$-recognizable. We have $\mathbf u_K(2n)=6^n$ and  $\mathbf u_K(2n+1)=3\cdot 6^n$ for all $n\in\N$. Then $\mathbf v_K(2n)\sim\frac{9}{5} 6^n$ and $\mathbf v_K(2n+1)\sim\frac{24}{5} 6^n$ $(n\to+\infty)$. From Theorem~\ref{the:main}, we obtain 
\[t_V(n)=\Theta\left(n^{\frac{\log(4)}{\log(\sqrt{6})}}\right).\]
\end{example}

\begin{example}
Consider the base $2$ numeration system, that is, the abstract numeration system built on 
$\mathcal L_2=\{\varepsilon\}\cup 1\{0,1\}^*$ with the natural order on the digits.
Let $X=\val_2(1^*0^*)=\{0,1,2,4,6,7,8,12,15,16,24,28,30,31,\ldots\}$. It is $2$-recognizable.
We have $\mathbf v_{1^*0^*}(n)=\binom{n+2}{2}$ for all $n\in\N$. From Theorem~\ref{the:main}, we obtain
\[t_X(n)=2^{(1+o(1))\sqrt{2n}}.\]
\end{example}

We can also use our main result to show that certain sets of integers
are not $S$-recognizable for any abstract numeration system $S$.

\begin{example}
Let $C = (C_n)_{n \geq 0}$ denote the set of \emph{Catalan numbers} \cite{A000108}:
i.e.,
\[
C_n = \frac{1}{n+1}{2n \choose n}.
\]
These numbers occur in many counting problems; for example, they count
the number of \emph{Dyck words} of length $2n$.  Asymptotically, we
have
\[
C_n \sim \frac{4^n}{n^{3/2}\sqrt{\pi}},
\]
which does not correspond to any of the forms described by
Theorem~\ref{the:main}.  Hence, for all $S$, the set $C$ is not
$S$-recognizable.
\end{example}

\section{Additional results}\label{sec:additional}

Ultimately periodic sets of integers play a special role. 
On the one hand, such infinite sets are coded by a finite amount of information. 
On the other hand, the famous theorem of Cobham asserts that these sets are 
the only ones that are recognizable in all integer base numeration systems. 
The following result shows that this property extends to abstract numeration systems.

\begin{theorem}\cite{LR}\label{the:ult_per}
Any ultimately periodic set is $S$-recognizable for all abstract numeration systems $S$ built on a regular language.
\end{theorem}

In \cite{CLR}, the latter result is extended to the multidimensional case.

Ultimately periodic sets are polynomial sets of degree $1$. The following example shows that there exists 
non-ultimately periodic polynomial sets of degree $1$ that are recognized in some abstract numeration systems.

\begin{example}
Consider the integer base $2$ numeration system. 
Let $t=(t_n)_{n\ge0}\in\{0,1\}^{\N}$ be the Thue-Morse sequence defined as follows: 
$t_n=(s_2(n)\mod 2)$ where $s_2(n)$ is the number of $1$'s in the $2$-representation $\rep_{S_2}(n)$ of $n$. 
Let $T\subseteq\N$ be the characteristic set of $t$: $n\in T$ if and only if $t_n=1$. 
It is well-known that $t$ is $2$-automatic, which is equivalent to the fact that $T$ is $S_2$-recognizable.
It is easily seen that the growth function $t_T(n)$ of $T$ is bounded by $2n$.
Furthermore, $T$ is not an ultimately periodic set.
\end{example}

\begin{theorem}\cite{Rigo}\label{the:construction}
Let $k$ be a positive integer and for all $i\in\{1,\ldots,k\}$, let $P_i\in\Q[x]$ be such that $P_i(\N)\subseteq\N$ 
and let $\alpha_i\in\N$. Define $f(x)=\sum_{i=1}^kP_i(x)\alpha_i^x$. 
Then an abstract numeration system $S$ built on a regular language $L$ such that $f(\N)$ is $S$-recognizable 
can be effectively provided. Furthermore, if $f(\N)$ is polynomial (resp. exponential), 
then the numeration language of the provided abstract numeration system  is polynomial (resp. exponential).
\end{theorem}

Let us recall the following result of Durand and Rigo.

\begin{theorem}\cite{DR}\label{the:DR}
Let $S$ be an abstract numeration system built on a polynomial regular language 
and let $T$ be an abstract numeration system built on an exponential regular language. 
If a subset of $\N$ is both $S$-recognizable and $T$-recognizable, then it is ultimately periodic.
\end{theorem}

\begin{proposition}\label{prop:poly}
Let $S$ be an abstract numeration system built on an exponential regular language. 
If $f\in\Q[x]$ is a polynomial of degree greater than $1$ such that $f(\N)\subseteq\N$, 
then the set $f(\N)$ is not $S$-recognizable.
\end{proposition}

\begin{proof}
It follows directly from Theorems \ref{the:construction} and \ref{the:DR}.
\end{proof}

The latter result has to be compared with Proposition~\ref{prop:kl}.
 
\section{Acknowledgments}
We would like to thank Michel Rigo for his initial questions that led
to this work and for his advice during the preparation of this paper.
We would also like to thank Jeffrey Shallit for some helpful
discussions.  We would especially like to thank our dear friend Anne
Lacroix for many interesting discussions during the course of this
work.

\bibliography{C:/Users/Emilie/Documents/bibliographie/bibliographie}
\bibliographystyle{alpha}

\end{document}